\documentclass[runningheads]{llncs}
\usepackage[T1]{fontenc}
\usepackage{graphicx}
\usepackage{amssymb}
\usepackage{amstext}
\usepackage{amsmath}

\usepackage{amsthm}
\usepackage{multicol}
\usepackage{mathrsfs}
\usepackage{mathtools}
\usepackage{nccmath}
\usepackage{setspace}
\usepackage{authblk}
\usepackage{fancybox}
\usepackage{dashbox}
\usepackage{enumitem}
\usepackage{mathcomp}

\theoremstyle{plain}
\newtheorem{df}{Definition}
\newtheorem{thm}{Theorem}

\newcommand{\dmid}{||}

\AtBeginDocument{
  \abovedisplayskip     =0.25\abovedisplayskip
  \abovedisplayshortskip=0.25\abovedisplayshortskip
  \belowdisplayskip     =0.25\belowdisplayskip
  \belowdisplayshortskip=0.25\belowdisplayshortskip
  \setlength{\jot}{2pt}
  \setlength{\topsep}{2pt}
  \setlength{\partopsep}{0pt}
  \setlength{\parsep}{0pt}
  \setlength{\itemsep}{0pt}}

\setlist[itemize]{topsep=1pt,parsep=0pt,partopsep=0pt,itemsep=1.5pt}
\setlist[description]{topsep=1pt,parsep=0pt,partopsep=0pt,itemsep=0pt}

\begin{document}
\title{On Translating Epistemic Operators in a Logic of Awareness}
%
%
\author{Yudai Kubono\orcidID{0000-0003-2617-8870}}
\authorrunning{Y. Kubono}
%
\institute{Graduate School of Science and Technology, Shizuoka University, Ohya, Shizuoka 422-8529, Japan 
\email{yudai.kubono@gmail.com}}
\maketitle              
\begin{abstract}
 Awareness-Based Indistinguishability Logic (henceforth, AIL) is an extension of Epistemic Logic by introducing the notion of awareness, distinguishing explicit knowledge from implicit knowledge. In this framework, each of these notions is represented by a modal operator. On the other hand, HMS models, developed in the economics literature, also provide a formalization of those notions. Nevertheless, the behavior of the epistemic operators in AIL within HMS models has yet to be explored. 
 In this paper, we define a transformation of an AIL model into an HMS model and then prove that a translation between the fragments of the language of AIL preserves truth under this transformation. As a result, we clarify the semantic role of an epistemic operator in AIL, which is induced by awareness and is essential to defining explicit knowledge, within HMS models. Furthermore, we demonstrate the differences in the implicit knowledge captured by AIL and HMS models. This work lays the groundwork for a comparative analysis between the model classes.
\end{abstract}
\section{Introduction}
Awareness-Based Indistinguishability Logic \cite{yudai2025-2} (henceforth, AIL) is an extension of Epistemic Logic (henceforth, EL), a variant of modal logic, by introducing the notion of \textit{awareness}, distinguishing \textit{explicit knowledge} from \textit{implicit knowledge}. In this framework, each of these notions, along with other epistemic notions to define explicit knowledge, is represented by a modal operator. AIL excludes specific undesirable formulae that are valid in the previous study \cite{fagin1988belief}. On the other hand, HMS models, first proposed in \cite{heifetz2006interactive}, also provide a formalization of awareness and explicit knowledge. Following the economic tradition, this structure is in the event-based approach \cite{fagin1995reasoning}, which employs operators on events represented as subsets of states (or worlds), unlike EL and AIL.

Nevertheless, the behavior of the epistemic operators in AIL within HMS models has yet to be explored. In this paper, we define a transformation of an AIL model into an HMS model and then prove that a translation between the fragments of the language of AIL preserves truth under this transformation. Via the translation, we clarify the semantic role of an epistemic operator in AIL, which is induced by awareness and is essential to defining explicit knowledge, within HMS models. Furthermore, we demonstrate the differences in the implicit knowledge captured by AIL and HMS models. 

\section{Preliminaries}
We introduce two languages that are fragments of the AIL language \cite{yudai2025-2}. We focus only on these fragments in this paper.
\begin{df}
 Let $\mathcal{P}$ be a countable set of atomic propositions and $\mathcal{G}$ a finite set of agents. The base language $\mathcal{L}_{PL}(\mathcal{P})$ is the set of formulae generated by the following grammar: 
 \begin{align*}
  &\mathcal{L}_{PL}(\mathcal{P}) \ni\varphi::= p \mid \neg\varphi \mid \varphi\wedge\varphi,
 \end{align*}
 where $p\in\mathcal{P}$. Other logical connectives $\vee$, $\to$, and $\leftrightarrow$ are defined in the usual manner. 
 Furthermore, the languages $\mathcal{L}^*_{HMS}(\mathcal{P,G})$ and $\mathcal{L}^*_{AIL}(\mathcal{P,G})$ are the sets of formulae generated by the following grammars: 
 \begin{align*}
  &\mathcal{L}^*_{HMS}(\mathcal{P,G}) \ni\psi::= \varphi \mid A_i \varphi \mid I_i\varphi;\\
  &\mathcal{L}^*_{AIL}(\mathcal{P,G}) \ni\psi::= \varphi \mid A_i \varphi \mid [\mathop{\approx}]_iI_i[\mathop{\approx}]_i\varphi,
 \end{align*}
 where $\varphi\in \mathcal{L}_{PL}(\mathcal{P})$ and $i\in\mathcal{G}$.
\end{df}
\noindent
We call the formulae in $\mathcal{L}^*_{AIL}$ AIL formulae. Furthermore, to avoid confusion, we sometimes attach superscripts indicating the languages to operators in formulae, such as $A^{HMS}_i$ and $A^{AIL}_i$. The operators $A$ and $I$ represent awareness and implicit knowledge, respectively. The other operator $[\mathop{\approx}]_i$ is a technical background one and is used to define explicit knowledge in \cite{yudai2025-2}.

\subsection{The Models in AIL}
A Model of AIL is presented in Kripke-style. We often call it an AIL model.
\begin{df}[\cite{yudai2025-2}]
  An \textit{epistemic model with awareness} $M$ is a tuple $\langle W, \{\mathop{\sim_i}, \allowbreak \mathscr{A}_i\}_{i\in\mathcal{G}}, V\rangle$, where:
     \begin{itemize}
      \item $W \text{ is a non-empty set of possible worlds}$;
      \item $\mathop{\sim_i}
      \text{ is an equivalence relation on }W$;
      \item $\mathscr{A}_i: W \to 2^{\mathcal{P}}$ is an awareness function satisfying that if $(w,v)\in\mathop{\sim_i}$, then $\mathscr{A}_i(w) = \mathscr{A}_i(v)$;
      \item $V : \mathcal{P} \to 2^{W}$ is a valuation.
     \end{itemize}
   \end{df}
\noindent
\begin{df}[\cite{yudai2025-2}]
  For each $i\in\mathcal{G}$, \textit{A-equivalence relation} $\mathop{\approx_i}$ on $W$ is defined by $(w,v) \in \mathop{\approx_i} \textit{\  iff,\ }$ $\mathscr{A}_i(w) = \mathscr{A}_i(v)$ and $w \in V(p) \textit{\ iff \ } v \in V(p) \text{ for every } p \in \mathscr{A}_i(w)$.
\end{df}
\noindent

We provide the satisfaction relation in our language $\mathcal{L}^*_{AIL}$. Let $At(\varphi)$ denote the set of atomic propositions occurring in $\varphi$.
\begin{df}[\cite{yudai2025-2}]
  For each epistemic model with awareness $M$ and possible world $w \in W$, the satisfaction relation $\vDash_{AIL^*}$ is given as follows: 
  \begin{align*}
    M,w \vDash_{AIL^*} p &\textit{\  iff  \ } w \in V(p);\\ 
    M,w \vDash_{AIL^*} \neg \varphi &\textit{\  iff  \ } M,w \nvDash_{AIL^*}\varphi;\\ 
    M,w \vDash_{AIL^*} \varphi\wedge\psi & \textit{\  iff  \ } M,w\vDash_{AIL^*}\varphi \text{, and } M,w\vDash_{AIL^*}\psi ; \\
    M,w \vDash_{AIL^*} A_{i} \varphi &\textit{\  iff  \ } At(\varphi) \subseteq \mathscr{A}_{i}(w);\\
    M,w \vDash_{AIL^*} I_i\varphi &\textit{\  iff  \ } M,v\vDash_{AIL^*} \varphi \text{ for all } v \text{ such that }(w,v)\in \mathop{\sim_i};\\
    M,w\vDash_{AIL^*} [\approx]_{i}I_i[\approx]_{i}\varphi &\textit{\  iff  \ } M,v\vDash_{AIL^*}\varphi \text{ for all } v \text{ such that }(w,v)\in \mathop{\approx_{i}}\circ \mathop{\sim_i}\circ \mathop{\approx_{i}},
  \end{align*}
  where $\mathop{\sim_i}\circ \mathop{\approx_{i}}$ is the sequential composition of $\mathop{\approx_{i}}$ and $\mathop{\sim_i}$.
\end{df}

\subsection{Implicit Knowledge-based HMS models}
In \cite{belardinelli2024implicit}, the authors defined \textit{implicit knowledge-based HMS models} (henceforth, iHMS models) based on the original HMS models \cite{heifetz2006interactive}. Since the models are structurally closer to AIL models, which makes them easier to compare, we focus on iHMS models in this paper.

An iHMS model is a tuple $\langle\{S_{\Phi}\}_{\Phi\subseteq\mathcal{P}},(r^{\Phi}_{\Psi})_{\Psi\subseteq\Phi\subseteq\mathcal{P}},\allowbreak(\Lambda_i,\alpha_i)_{i\in\mathcal{G}},v\rangle$ consisting of a non-empty collection of non-empty disjoint state spaces; a family of projections; implicit possibility correspondences; awareness functions; and a valuation. We do not reproduce the formal definitions here; for the detailed definition, the assumptions required for these models, and the definition of operators on events, see \cite{belardinelli2024implicit}. We use iHMS models as the target models into which $\mathcal{L}^*_{AIL}$ is translated. Furthermore, we introduce our own notations such as $e_{A_i}(E)$ and $e_{I_i}(E)$ for the awareness and implicit knowledge operators, where $E$ is an event. 

\section{Interpreting AIL Formulae in HMS Models via a Translation}
We define the HMS-transform of an AIL model. In this transform, partitions by A-equivalence relations and mappings between them serve as state-spaces and projections in iHMS models.
\begin{df}
  Let $M$ be an epistemic model with awareness such that $\mathscr{A}_i(w) = \mathscr{A}_i(v)$ for every $i\in\mathcal{G}$ and for all $w,v\in W$. The HMS-transform of $M$ is a tuple $t^{M\rightarrow}(M) = \langle \{W_{\Phi}\}_{\Phi\subseteq\mathcal{P}},(m^{\Phi}_{\Psi})_{\Psi\subseteq\Phi\subseteq\mathcal{P}}, (\Lambda^*_i, \alpha^*_i)_{i\in\mathcal{G}},v^*\rangle$: 
  \begin{itemize}
    \item $W_{\Phi} \coloneqq W/\approx_{\Phi}$, where $\approx_{\Phi}$ is defined by $(w,v) \in \mathop{\approx_{\Phi}} \text{\  iff,\ }$ $w \in V(p) \text{\ iff \ } v \in V(p)$ for every $p \in \Phi$. We denote $\bigcup_{\Phi\subseteq\mathcal{P}}W_{\Phi}$ as $W^*$, an element of $W^*$ as $w^*$, and an equivalence class including $w$ in $W_{\Phi}$ as $[w]_{\Phi}$;
    \item $m^{\Phi}_{\Psi}: W_{\Phi}\to W_{\Psi}$ is defined by $m^{\Phi}_{\Psi}([w]_{\Phi}) = [w]_{\Psi}$.
    \item $\Lambda^*_i: W^* \to 2^{W^*}\setminus \{\emptyset\}$ is defined by $[v]_{\Phi}\in\Lambda_i([w]_{\Phi}) \text{ iff } (w',v')\in\mathop{\sim}_i$ for some $w'\in[w]_{\Phi} \text{ and } v'\in[v]_{\Phi}$, $\text{if }\Phi = \mathcal{P}$; and by $\Lambda_i([w]_{\Phi}) = \{m^{\mathcal{P}}_{\Phi}(w^*)\mid w^*\in \Lambda_i([w]_{\mathcal{P}})\}$ otherwise.
    \item $\alpha^*_i: W^*\to \{W_{\Phi}\}_{\Phi\subseteq\mathcal{P}}$ is defined by $W_{\Psi} = \alpha^*_i([w]_{\Upsilon})$ iff $\Psi = \mathscr{A}_i(w) \cap \Upsilon$.
    \item $v^*: \mathcal{P}\to 2^{W^*}$ is defined by $v^*(p) = \bigcup_{\{p\}\subseteq\Phi\subseteq\mathcal{P}} \{[w]_{\Phi} \mid w\in V(p)\}$. 
  \end{itemize}
\end{df}
\noindent
An iHMS model has a space associated with a set of atomic propositions that does not correspond to any agent's awareness set, unlike a partition by an A-equivalence relation. Hence, we define each space by partitioning $W$ with $\mathop{\approx}_{\Phi}$ for each $\Phi\subseteq \mathcal{P}$.


We now show that any HMS-transform is an iHMS model. 
\begin{lemma}
  For any AIL model $M$, the HMS-transform $t^{M\rightarrow}(M)$ is an implicit knowledge-based HMS model.
\end{lemma}
\begin{proof}
  \vspace{-3pt}
  The proof follows from the construction of the HMS-transform. 
\end{proof}

HMS models, including iHMS models, are structures in the event-based approach \cite{fagin1995reasoning}, which focuses on events and dispenses with formulae. To define a truth-preserving translation of the formulae, we need to introduce the satisfaction relation for HMS-transforms.
\begin{df}
  For any HMS-transforms $t^{M\rightarrow}(M)$ and $w^*\in W^*$, the satisfaction relation $\vDash_{HMS}$ is given as follows: 
  \begin{align*}
    t^{M\rightarrow}(M), w^* \vDash_{HMS} p &\textit{\  iff  \ } w^* \in v(p) ;\\[-3pt] 
    t^{M\rightarrow}(M), w^* \vDash_{HMS} \neg \varphi &\textit{\  iff  \ } w^*\in \neg\dmid\varphi\dmid;\\[-3pt]
    t^{M\rightarrow}(M), w^* \vDash_{HMS} \varphi\wedge\psi &\textit{\  iff  \ } w^*\in\dmid\varphi\dmid\cap\dmid\psi\dmid; \\[-3pt] 
    t^{M\rightarrow}(M), w^* \vDash_{HMS} A_{i} \varphi &\textit{\  iff  \ } w^*\in e_{A_i}(\dmid\varphi\dmid);\\[-3pt]
    t^{M\rightarrow}(M), w^* \vDash_{HMS} I_i\varphi &\textit{\  iff  \ } w^*\in e_{I_i}(\dmid\varphi\dmid).
  \end{align*}
  where $\dmid\varphi\dmid \coloneqq \{w^*\mid t^{M\rightarrow}(M), w^* \vDash_{HMS}\varphi\}$ for all $\varphi\in\mathcal{L}^*_{HMS}$.
\end{df}

In order for the satisfaction relation to be well-defined, we must prove that the truth set $\dmid \varphi\dmid$ of every formula is an event in HMS-transforms.
\begin{lemma}
  For every $\varphi\in\mathcal{L}^*_{HMS}$ and HMS-transform $t^{M\to}(M)$, $\dmid\varphi\dmid$ is a $W_{At(\varphi)}$-based event.
\end{lemma}
\begin{proof}
\vspace{-3pt}
  The proof is by induction on the structures of $\varphi$. It is sufficient to show that there exists $B\subseteq W_{At(\varphi)}$ such that $\dmid\varphi\dmid = \bigcup_{At(\varphi)\subseteq\Phi\subseteq\mathcal{P}}(m^{\Phi}_{At(\varphi)})^{-1}(B)$ for each case, where $m^{\Phi}_{\Psi}(B)\coloneqq\{m^{\Phi}_{\Psi}([w]_{\Phi})\mid [w]_{\Phi}\in B\}$ for $B\subseteq W_{\Phi}$. Such $B$ is obtained by taking $\{w^*\in W_{At(\psi)}\mid W_{At(\psi)}\leq\alpha^*_i(w^*)\}$ for $A_i\psi$, and by taking $\bigcup_{w^*\in W^*}\{\Lambda^*_i(w^*)\mid \Lambda^*_i(w^*)\subseteq \dmid \psi \dmid \cap W_{At(\psi)}\}$ for $I_i\psi$.
\end{proof}

The next task is to define a translation of formulae in $\mathcal{L}^*_{AIL}$ into $\mathcal{L}^*_{HMS}$. 
\begin{df}
  The translation $t: \mathcal{L}^*_{AIL}\to\mathcal{L}^*_{HMS}$ is defined as follows:
  \begin{align*}
    & t(p) = p; &  & t(\neg \varphi) = \neg t(\varphi);\\
    & t(\varphi\wedge\psi) = t(\varphi)\wedge t(\psi); &  &  t(A^{AIL}_i\varphi) = A^{HMS}_i t(\varphi);\\
    & t([\mathop{\approx}]_i I^{AIL}_i[\mathop{\approx}]_i\varphi) = I^{HMS}_i(t(\varphi)).
  \end{align*}
\end{df}
\noindent
The translation replaces each occurrence of the $[\mathop{\approx}]_i I_i [\mathop{\approx}]_i$ with the $I_i$ operator.

Finally, we prove that the translation is truth-preserving between an AIL model and its HMS-transform. For every $\varphi \in \mathcal{L}^*_{AIL}$, we say that $\varphi$ satisfies \textit{A-condition} at $w\in W$ precisely when, if $\varphi$ is of the form $A_i\psi$ or $[\mathop{\approx}]_i I_i [\mathop{\approx}]_i \psi$ then $At(\psi) = \mathscr{A}_i(w)$.
\begin{thm}
  Let $M$ be an epistemic model with awareness such that $\mathscr{A}_i(w) = \mathscr{A}_i(v)$ for every $i$ and for all $w,v\in W$, $t^{M\rightarrow}(M)$ be its HMS-transform, and $w\in W$. For every $\varphi\in\mathcal{L}^*_{AIL}$, if $\varphi$ satisfies A-condition at $w$, then $M,w\vDash_{AIL} \varphi \textit{\  iff \ }$ $t^{M\rightarrow}(M),[w]_{At(\varphi)}\vDash_{HMS}t(\varphi)$.
\end{thm}
\begin{proof}
  \vspace{-3pt}
  The proof is by induction on the structure of $\varphi$, using Lemma 2. For the case of $A_i\psi$, this immediately follows from A-condition. For the case of $[\mathop{\approx}]_i I_i [\mathop{\approx}]_i\psi$, the proof follows from that $[w]_{At([\mathop{\approx}]_i I_i [\mathop{\approx}]_i\psi)}$ is defined by $\mathop{\approx}_i$ and that $\Lambda^*_i$ is defined by $\mathop{\sim}_i$. 
\end{proof}
The result implies that the $[\mathop{\approx}]_i I_i [\mathop{\approx}]_i$ operator in an AIL model is reduced to the $I_i$ operator at the subjective state of $i$ in the corresponding iHMS model under some assumptions. 
This is ascribed to the structure of iHMS models, in which the implicit knowledge is evaluated within states that appear under a vocabulary restriction imposed by an agent's awareness. In AIL, in contrast, the knowledge is defined independently of an agent's awareness, with each agent's subjective states being expressed through a relation $\mathop{\approx}_i$. In other words, while iHMS models internalize an agent's viewpoint into a state (or an evaluation point), AIL models externalize it as a relation $\mathop{\approx}_i$ between worlds. Hence, the $[\mathop{\approx}]_i$ operator serves as an indicator of an agent's viewpoint within iHMS models, which means from whose viewpoint a formula is evaluated.

\section{Conclusion}
In this paper, we defined the transformation of an AIL model into an HMS model and then proved that the translation between the fragments of the AIL language preserves truth within this transformation. Thus, we demonstrated that the $[\mathop{\approx}]_i I_i [\mathop{\approx}]_i$ operator in an AIL model is reduced to the $I_i$ operator at the subjective state of $i$ in the corresponding iHMS model under some assumptions. This work lays the groundwork for a comparative analysis between the two model classes.

Our result implies that information about whose viewpoint a state is from, which an evaluation point itself has in iHMS models, is externalized by the $[\mathop{\approx}]_i$ operator in AIL; the $[\mathop{\approx}]_i$ operator serves as an indicator of an agent's viewpoint within iHMS models. Furthermore, our truth-preserving translation shows that iHMS models capture a different notion of implicit knowledge than AIL does.

In iHMS models, explicit knowledge coincides with the intersection of implicit knowledge and awareness \cite{belardinelli2024implicit}. It is left for future work to explore the formal comparison between explicit knowledge in iHMS models and that in AIL.

\subsubsection{\ackname}
We thank Satoshi Tojo and Nobu-Yuki Suzuki for their helpful comments.

\bibliography{refs}
\end{document}